\documentclass{article}

\usepackage[english]{babel}

\usepackage[letterpaper,top=2cm,bottom=2cm,left=3cm,right=3cm,marginparwidth=1.75cm]{geometry}

\usepackage{amsmath}
\usepackage{graphicx}
\usepackage[colorlinks=true, allcolors=blue]{hyperref}

\usepackage{amsfonts}
\usepackage{amsmath}
\usepackage{amsfonts}
\usepackage{amsthm}

\usepackage{subfigure}
\usepackage{url}
\usepackage{authblk}

\newtheorem{lemma}{Lemma}

\title{Phylogenetic Inference under the Balanced Minimum Evolution Criterion via Semidefinite Programming}

\author[1,$\dagger$]{P. Skums}

\affil[1]{School of Computing, University of Connecticut, Storrs, CT, USA}
\affil[$\dagger$]{Email: pavel.skums@uconn.edu}

\date{}

\begin{document}
\maketitle

\begin{abstract}
In this study, we investigate the application of Semidefinite Programming (SDP) to phylogenetics. SDP is a powerful optimization framework that seeks to optimize a linear objective function over the cone of positive semidefinite matrices. As a convex optimization problem, SDP generalizes linear programming and provides tight relaxations for many combinatorial optimization problems. However, despite its many applications, SDP remains largely unused in computational biology.

We argue that SDP relaxations are particularly well suited for phylogenetic inference. As a proof of concept, we focus on the Balanced Minimum Evolution (BME) problem, a widely used model in distance-based phylogenetics. We propose an algorithm combining an SDP relaxation with a rounding scheme that iteratively converts relaxed solutions into valid tree topologies. Experiments on simulated and empirical datasets show that the method enables accurate phylogenetic reconstruction.

The code is available at \url{https://github.com/compbel/SDPTree}.
\end{abstract}

\section{Introduction}

Phylogenetics plays a fundamental role in modeling evolutionary relationships and population dynamics across a wide range of domains, including molecular epidemiology, cancer genomics, developmental biology and genomic anthropology. A common feature of almost all phylogenetic models is that the corresponding optimization problems are NP-hard \cite{day1987computational,chor2005maximum}. Consequently, the most widely used algorithmic approaches to phylogenetic inference have relied on heuristic methods based on local search or Markov chain Monte Carlo (MCMC) sampling. 

In these approaches, the tree space is explored iteratively by modifying the current tree through subtree rearrangement operations that alter the tree topology. This strategy underlies many widely used phylogenetic
inference tools and has proven effective in practice, making such heuristics the backbone of modern
phylogenetic analysis.

Nevertheless, the ability of these methods to reach optimal solutions or converge to the true posterior distribution remains uncertain, even for moderately sized datasets \cite{spirin2024phylobench,zhou2018evaluating,liu2024influence,fabreti2022convergence}. This difficulty largely arises from the superexponential size and ruggedness of the solution space, together with a strong dependence on the initial solution. 

Given the aforementioned challenges, alternative approaches to phylogenetic inference have been explored. The most prominent of them is Integer Linear Programming (ILP) \cite{gusfield2019integer}. ILP offers several advantages, including their ability to find optimal solutions, relative ease of expressing many problems in ILP form, and the availability of efficient commercial and open-source solvers. However, ILP-based methods face significant scalability limitations, which make them practical only for small-sized problems. One of major reasons for this is that most ILP formulations require large numbers of auxiliary variables simply to ensure that the resulting solution encodes a valid tree with the desired properties. For the same reason, straightforward approximation strategies, such as linear programming (LP) relaxation followed by rounding, are often ineffective, as satisfying all constraints after rounding is exceedingly difficult.

In this study, we introduce an alternative mathematical programming approach to phylogenetic inference based on {\it Semidefinite Programming (SDP)}. SDP is a powerful optimization methodology that involves minimizing a linear objective function subject to linear constraints over the cone of positive semidefinite matrices \cite{vandenberghe1996semidefinite}. It is a nonlinear convex optimization problem that can be solved in polynomial time using interior-point methods. SDP generalizes LP and provides tight relaxations for a wide range of combinatorial optimization problems \cite{lovasz2003semidefinite}. It has been successfully applied in many fields, ranging from engineering \cite{wolkowicz2012handbook} to quantum computing \cite{skrzypczyk2023semidefinite}. However, aside from a few isolated applications in haplotyping \cite{kalpakis2005haplotype}, protein structure determination \cite{alipanahi2013determining}, and protein docking \cite{marcia2007global}, the use of SDP in computational biology remains highly limited.

Here, we argue that SDP relaxations are particularly well suited for phylogenetic inference. A phylogenetic tree can naturally be viewed as a hierarchy of cuts. Semidefinite relaxations have proven especially effective for cut-based optimization problems; a well-known example is the celebrated Goemans–Williamson $0.878$-approximation algorithm for the MaxCut problem \cite{goemans1995improved}. Motivated by this connection, we suggest that SDP provides a convenient framework for representing tree-like structures within convex relaxations, which can then be efficiently rounded to obtain discrete tree solutions. 

As a proof of concept, we consider the phylogenetic inference problem under the Balanced Minimum Evolution (BME) model. This model relies on a specific branch-length estimation approach introduced in \cite{pauplin2000direct}. The resulting inference problem is equivalent to identifying the phylogeny that minimizes a weighted least-squares discrepancy between a given dissimilarity matrix $D$ and the corresponding tree distances. The weights decrease exponentially with the number of internal nodes separating taxa, reflecting the increasing variance associated with larger evolutionary distances \cite{desper2004theoretical}.

The BME phylogenetic model is one of the most widely used approaches in distance-based phylogenetics. It has been extensively studied theoretically \cite{desper2002fast,desper2004theoretical,pauplin2000direct,semple2004cyclic,forcey2016facets}, implemented in widely used software tools \cite{lefort2015fastme}, and applied in numerous empirical studies \cite{wick2021trycycler,edger2019origin,librado2021origins}. The model also has several important features:  it is statistically consistent \cite{semple2004cyclic,desper2004theoretical}; the classical Neighbor-Joining (NJ) algorithm is essentially as a greedy heuristic for optimizing the BME objective \cite{gascuel2006neighbor}; and in practice it often achieves higher accuracy than alternative phylogenetic methods on real data \cite{spirin2024phylobench}.

BME tree inference problem is NP-hard and, moreover, cannot be approximated within a factor of $c^n$ for a constant $c>1$ \cite{fiorini2012approximating,frohn2021approximability}. Consequently, a variety of exact and heuristic algorithmic approaches have been explored, including local search  \cite{desper2002fast,gasparin2023evolution,lefort2015fastme} and ILP \cite{forcey2016facets,catanzaro2020balanced,catanzaro2012balanced,catanzaro2026optimizing,catanzaro2025new}.

 In this paper, we present a novel algorithmic approach to BME problem based on solving an SDP relaxation followed by rounding scheme for transforming relaxed solutions to actual tree solutions. The relaxation preserves the dyadic weighting structure of BME, captures hierarchical nesting of clades through semidefinite constraints, and remains convex, enabling efficient solution via interior-point methods. Using simulated and empirical data, we demonstrate that the resulting algorithm, which we call {\it SDPTree}, allows for accurate and scalable phylogenetic reconstruction. The proposed scheme is general enough to be extendable to other phylogenetic inference problems.

 \section{Methods}

 \subsection{Matrix formulation of the BME problem}

Consider $n$ taxa together with their $n \times n$ dissimilarity matrix $D=(d_{ij})_{i,j=1}^n$. Let $\mathcal{T}_n$ denote the set of all full binary phylogenetic trees (i.e., trees in which every vertex has degree either $1$ or $3$) with leaf set $\{1,\ldots,n\}$. The Balanced Minimum Evolution (BME) problem can be formulated in several equivalent ways \cite{desper2004theoretical,pauplin2000direct}. One formulation that is particularly convenient for algorithmic purposes can be written compactly as

\begin{equation}\label{eq:BMEobj}
\min_{T\in \mathcal{T}_n} F(T) = \sum_{1 \le i,j \le n} d_{ij}2^{-\delta_{ij}},
\end{equation}

\noindent
where $\delta_{ij}$ denotes the topological distance between leaves $i$ and $j$, i.e., the number of edges on the unique $(i,j)$-path in $T$.

For our purposes, we derive an equivalent formulation in which the objective is written in matrix form. The formulation is most
naturally defined for a rooted tree $T$; the unrooted objective can be obtained from the rooted one by a simple adjustment (see below).

We denote by $\rho_u$ the depth (distance from the root) of
a vertex $u\in V(T)$. Let $\sigma_{ij}$ be the depth of the
lowest common ancestor of leaves $i,j\in [n]$, and let
$K=\max_{i\in[n]}\rho_i$ denote the height of the tree.
Then $\delta_{ij}=\rho_i+\rho_j-2\sigma_{ij}$.

Using the notation
$\zeta=(2^{-\rho_i})_{i=1}^n$
and
$\Sigma=(4^{\sigma_{ij}})_{i,j=1}^n$,
the formulation (\ref{eq:BMEobj}) can be written as

\begin{equation}\label{eq:BMEobjMatr}
\min_{T\in \mathcal{T}_n} F(T)
= \langle D, \zeta\zeta^{\top}\circ\Sigma\rangle.
\end{equation}

\noindent
Here $\langle A,B\rangle$ denotes the matrix inner product and
$\circ$ the Hadamard (element-wise) product. The variables
$\zeta_i$ will be referred to as {\it dyadic weights}.

Denote by $\mathbb{S}^n$ the set of real symmetric $n\times n$
matrices. A matrix $M\in\mathbb{S}^n$ is {\it positive semidefinite}
(PSD, denoted $M\succeq0$) iff any of the following equivalent
conditions holds: (a) $w^{\top}Mw\ge0$ for all
$w\in\mathbb{R}^n$; (b) $M$ is a {\it Gram matrix},
i.e. $M=L^{\top}L$ for some $L\in\mathbb{R}^{m\times n}$;
(c) all eigenvalues of $M$ are nonnegative.

Consider the matrices $B^{(0)},\dots,B^{(K)}$ defined by

\begin{equation}
B^{(k)}_{ij} =
\begin{cases}
1, & \sigma_{ij},\rho_i,\rho_j \ge k, \\
0, & \text{otherwise}.
\end{cases}
\end{equation}

That is, if $R_k=\{u\in V(T):\rho_u=k\}$ is the set of vertices
at depth $k$, then each $u\in R_k$ defines a cluster $L_u$ of
descendant leaves; $B^{(k)}_{ij}$ is the indicator of whether
$i$ and $j$ belong to the same cluster at level $k$. In particular,
$B^{(0)}=\mathbf{1}_n\mathbf{1}_n^{\top}$ is the all-ones matrix.


\begin{lemma}\label{lem:PSD}
(i) $B^{(k)}\succeq0$ for $k=0,\dots,K$.

(ii) $\Sigma=\sum_{k=0}^K \beta_k B^{(k)}\succeq0$, where
$\beta_0=1$ and $\beta_k=3\cdot4^{k-1}$ for $k\ge1$.
\end{lemma}

\begin{proof}
    
For each $k$, let $U^{(k)}$ be the $|R_k|\times n$ matrix whose
rows are the indicator vectors of clusters $L_u$, $u\in R_k$.
Then $B^{(k)}=(U^{(k)})^{\top}U^{(k)}$, and hence $B^{(k)}\succeq0$.
    
 Relation (ii) follows from the observation that for any
$i,j\in[n]$ the entry $\Sigma_{ij}$ admits the telescoping
representation
\begin{equation}\label{eq:PSDU}
\Sigma_{ij}=1+\sum_{k=1}^{K}(4^k-4^{k-1})B_{ij}^{(k)}.
\end{equation}
Thus $\Sigma$ is a positive linear combination of PSD
matrices and hence PSD.
\end{proof}

Using Lemma \ref{lem:PSD}, the rooted BME problem can be formulated as
\begin{equation}\label{eq:matrformBME}
\min_{T\in \mathcal{T}_n} F(T)
= \sum_{k=0}^H \beta_k\langle D,\zeta\zeta^{\top}\circ B^{(k)} \rangle .
\end{equation}

\noindent
Since $B^{(k)}\succeq0$ and $\zeta\zeta^{\top}\succeq0$,
their Hadamard product is also PSD by the
Schur product theorem; hence
$\zeta\zeta^{\top}\circ B^{(k)}\succeq0$.

For the unrooted case, note that distances in $T$ and in any rooting $T'$
coincide unless $i$ and $j$ lie on different sides of the root,
in which case $2^{-\delta_{ij}}=2\cdot2^{-\delta'_{ij}}$.
Thus the unrooted objective is obtained by setting
$\beta_0 = \beta_1 =2$ and using (\ref{eq:matrformBME}).

\subsection{SDP relaxation of the BME problem}

Recall that SDP is the problem of minimizing a linear function
$\langle C,X\rangle$ of a PSD matrix variable $X\succeq0$
subject to linear constraints $\langle A_i,X\rangle\le b_i$,
$i=1,\dots,m$. Thus (\ref{eq:matrformBME}) is not a valid SDP
objective, since both $\zeta$ and $B^{(k)}$ are variables.
Therefore we introduce additional variables relaxing the products $\zeta\zeta^{\top}\circ B^{(k)}$.

\paragraph{Variables.}

Let $K\ge1$ be a prescribed upper bound on the tree height
(in the simplest case $K=n-1$). We introduce the following variables:

\begin{itemize}
\item $P\in[0,1]^{n\times K}$: relaxed depth assignments.
Each row represents a distribution over depths
$\rho_i\in\{1,\dots,K\}$, i.e.\ $P_{i,d}=\Pr(\rho_i=d)$.

\item $Z\in\mathbb{S}^n$: a matrix variable relaxing
the rank-one matrix $\zeta\zeta^{\top}$.

\item $Y^{(k)}\in\mathbb{S}^n$, $k=0,\dots,K$:
level-wise PSD matrices approximating
$\zeta\zeta^{\top}\circ B^{(k)}$ and encoding the
hierarchical partition of taxa induced by the tree.
\end{itemize}

To simplify notation, we also introduce auxiliary variables expressible in terms of the main variables. These are not essential for the optimization itself but make the constraints more compact and easier to state.

\begin{itemize}
    \item $z\in\mathbb{R}^n$ and $s\in\mathbb{R}^n$: dyadic leaf
    weights and their second moments, approximating
    $z_i\approx\zeta_i=2^{-\rho_i}$ and $s_i\approx\zeta_i^2$.
    They are coupled with $P$ via
    \begin{align}\label{constr:zdef}
    z_i &= \mathbb{E}[2^{-\rho_i}]
    = \sum_{d=1}^K 2^{-d} P_{i,d}, &
    s_i &= \mathbb{E}[2^{-2\rho_i}]
    = \sum_{d=1}^K 2^{-2d} P_{i,d}.
    \end{align}      

    \item $z^{(k)}, s^{(k)}, q^{(k)} \in \mathbb{R}^n$:
    truncated moments restricted to tree levels $\ge k$:
    \begin{align}
    z_i^{(k)} &= \mathbb{E}\!\left[2^{-\rho_i}\mathbb{I}\{\rho_i\ge k\}\right]
              = \sum_{d=k}^K 2^{-d}P_{i,d}, \label{constr:zk}\\
    s_i^{(k)} &= \mathbb{E}\!\left[2^{-2\rho_i}\mathbb{I}\{\rho_i\ge k\}\right]
              = \sum_{d=k}^K 2^{-2d}P_{i,d}, \label{constr:sk}\\
    q_i^{(k)} &= \mathbb{E}\!\left[\mathbb{I}\{\rho_i\ge k\}\right]
              = \sum_{d=k}^K P_{i,d}. \label{constr:qk}
    \end{align}
\end{itemize}

\paragraph{Probability distribution constraint:}

\begin{align}
P\mathbf{1}_K &= \mathbf{1}_n \label{constr:onedepth} 
\end{align}

\paragraph{Kraft constraints:}

\begin{align}
z^{\top}\mathbf{1}_n &= 1 \label{constr:kraft}\\
Z\mathbf{1} &= z, \label{constr:Z-rowsum} \\
Y^{(k)} \mathbf{1}_n &= 2^{-k} z^{(k)}, \qquad k=0,...,K \label{constr:kraftlevel}
\end{align}

Constraint (\ref{constr:kraft}) relaxes the Kraft equality
$\sum_{i=1}^n \zeta_i=1$ for leaf depths in a full binary tree
\cite{catanzaro2012balanced,cover1999elements}.
Constraint (\ref{constr:Z-rowsum}) relaxes its corollary
$\sum_{j}\zeta_i\zeta_j=\zeta_i$ by replacing
$\zeta_i\zeta_j$ and $\zeta_i$ with $Z_{ij}$ and $z_i$.
Constraint (\ref{constr:kraftlevel}) is the subtree version,
$\sum_j B^{(k)}_{ij}\zeta_j=2^{-k}$, multiplied by $\zeta_i$
and relaxed via $B^{(k)}_{ij}\zeta_i\zeta_j\mapsto Y^{(k)}_{ij}$.

\paragraph{PSD and non-negativity constraints}

\begin{align}
Y^{(k)},Z &\succeq 0, \label{eq:Y-psd}\\
Y^{(k)}_{ij}, Z_{ij} &\ge 0 \qquad i,j\in [n]. \label{eq:Y-nonneg}
\end{align}

\paragraph{Diagonal constraints}

\begin{align}
\operatorname{diag}(Z) &= s, \label{constr:diagZ} \\
\operatorname{diag}(Y^{(k)}) &= s^{(k)}, \qquad k\in[K]. \label{constr:diagY}
\end{align}

\noindent
These relax the identities
$(\zeta\zeta^{\top})_{ii}=2^{-2\rho_i}$ and
$(\zeta\zeta^{\top}\!\circ B^{(k)})_{ii}
=2^{-2\rho_i}\mathbf{1}\{\rho_i\ge k\}$.

\paragraph{Cluster nestedness}
\begin{equation}
Y^{(k+1)}_{ij} \le Y^{(k)}_{ij}
\qquad i,j\in [n],\ k=0,\dots,K-1.
\end{equation}

This relaxes the requirement that clusters form a nested hierarchy: if two leaves are separated at level $k$, they
remain separated at all deeper levels.

\paragraph{Shor lifting}

\begin{align}
\begin{pmatrix}
1 & z^\top \\
z & Z
\end{pmatrix} &\succeq 0. \label{constr:shor} \\
\begin{pmatrix}
q_i^{(k)} & z_i^{(k)}\\
z_i^{(k)} & s_i^{(k)}
\end{pmatrix} &\succeq 0, \qquad i \in [n],\ k=0,\dots,K. \label{constr:2x2}
\end{align}

Constraint (\ref{constr:shor}) is the \emph{Shor semidefinite lifting}
\cite{ben2001lectures}, equivalent via the Schur complement
\cite{boyd2004convex} to $Z-zz^\top\succeq0$, and thus a convex
relaxation of $Z=zz^\top$. Constraint (\ref{constr:2x2}) implies
$(z_i^{(k)})^2\le q_i^{(k)}s_i^{(k)}$, the Cauchy--Schwarz
inequality applied to (\ref{constr:zk})--(\ref{constr:qk}).

\paragraph{McCormick envelope}
\begin{align}
Z_{ij} &\le  \frac{1}{2}z_i + \frac{1}{2^L} z_j - \frac{1}{2^{L+1}},\label{constr:mccorm1}\\
Z_{ij} &\le \frac{1}{2^L} z_i + \frac{1}{2} z_j - \frac{1}{2^{L+1}},\label{constr:mccorm2}\\
Z_{ij} &\ge \frac{1}{2^L} z_i + \frac{1}{2^L} z_j - \frac{1}{2^{2L}},\label{constr:mccorm3}\\
Z_{ij} &\ge \frac{1}{2} z_i + \frac{1}{2} z_j - \frac{1}{4}, 1\leq i<j\leq n \label{constr:mccorm4}
\end{align}

Constraints (\ref{constr:mccorm1})--(\ref{constr:mccorm4}) form the {\it McCormick envelope} \cite{mccormick1976computability},
a standard convex relaxation for bilinear relations such as $Z_{ij}=z_i z_j$.

\paragraph{Additional coupling constraints}

\begin{align}
Y^{(0)} &= Z, \label{constr:Z}\\
Y^{(k)}_{ij} &\le Z_{ij},
\qquad i,j\in[n],\ k=0,\dots,K, \label{constr:ineq1}\\
Z_{ij} &\le z_i,\qquad
Z_{ij}\le z_j,
\qquad i,j\in[n]. \label{constr:ineq2}
\end{align}

\noindent
Here, (\ref{constr:Z}) relaxes
$\zeta\zeta^{\top}\!\circ B^{(0)}=\zeta\zeta^{\top}$
(since $B^{(0)}=\mathbf{1}_n\mathbf{1}_n^{\top}$),
while (\ref{constr:ineq1}) reflects
$Y_{ij}^{(k)}\approx \zeta_i\zeta_j B^{(k)}_{ij}
\le \zeta_i\zeta_j \approx Z_{ij}$.

In summary, the SDP relaxation of the BME problem is

\begin{equation}\label{obj:SDP}
\min \sum_{k=0}^K \beta_k \langle D, Y^{(k)} \rangle
\quad \text{s.t. } (\ref{constr:zdef})\text{--}(\ref{constr:ineq2}).
\end{equation}

\subsection{Rounding the relaxation}

After solving the SDP relaxation
(\ref{obj:SDP}), the fractional solution is converted into a binary tree. Rather than constructing the tree in a single step, we proceed agglomeratively. At each iteration we identify, from the SDP solution, the pair of taxa most likely to form a cherry, merge them into their parent,
update the objective accordingly, and repeat the SDP + cherry detection procedure for $n-1$ taxa until the full tree $T$ is obtained. A similar approach for an ILP relaxation was independently proposed
in \cite{catanzaro2025new}.

To estimate likely cherries we use two approaches, further referred to as {\it separability} and {\it profile} rounding heuristics, respectively (or, for short, {\it s-rounding} and {\it p-rounding}).
In the separability approach, we recover relaxed clade co-membership matrices
$B^{(k)}$ via
$B^{(k)}_{ij}=Y^{(k)}_{ij}/\sqrt{Z_{ii}Z_{jj}}$. We then define
$C_{ij}=|\{k: B^{(k)}_{ij}=\max_{u\ne v} B^{(k)}_{uv}\}|$,
i.e. $C_{ij}$ counts the number of levels at which $i$ and $j$
are strongly inseparable. Finally, we find the maximum-weight matching $(i_1,j_1),...,(i_{L'},j_{L'})$ of the size $L'\leq L$ in the weighted graph $G_C$ with weights $C$, and merge the pairs $(i_l,j_l)$, $l\in [L']$.

Alternatively, we consider
$\Delta=\sum_{k=0}^K \beta_k Y^{(k)}$.
By (\ref{eq:PSDU}), for binary solutions
\begin{equation}
\Delta_{ij}
= \sum_{k=0}^K \beta_k \zeta_i\zeta_j B_{ij}^{(k)}
= \zeta_i\zeta_j \Sigma_{ij}
= 2^{-\delta_{ij}},
\end{equation}

\noindent
so $\Delta_{ij}$ relaxes $2^{-\delta_{ij}}$.
We therefore consider the weighted graph $G_{\Delta}$ with edges weights defined as $w_{ij} = \|\Delta_{i,:}-\Delta_{j,:}\|_1$ and merge the pairs $(i_l,j_l)$, $l\in [L]$ forming the minimum-weight $L$-matching in $G_{\Delta}$. In other words, we are merging pairs whose
distance profiles to the remaining leaves are most similar.

The selected pair $(i,j)$ is replaced by its parent $p_{ij}$, and the dissimilarity matrix $D$ is updated according to the
BME objective by replacing the $i$th and $j$th rows and columns with their average. 

The matching size $L$ is the algorithm parameter with the default value $L=1$; In our implementation minimum or maximum $L$-matchings are estimated using the simple greedy algorithm. In practice, we apply both separability and profile rounding heuristics and
select the better of the two resulting trees, which is then refined using Subtree Pruning and Regrafting (SPR) local search.

The full BME inference algorithm called SDPTree was implemented in Matlab R2024b (MathWorks Inc.). The SDP relaxation was solved using MOSEK~11.1 (Mosek ApS) via the YALMIP toolbox \cite{Lofberg2004}.

 \section{Results}

 \subsection{Data and algorithms}

The proposed algorithm was validated and benchmarked on simulated
and real data. In all experiments, the input consisted of randomly generated dissimilarity matrices
$D\in\mathbb{R}^{n\times n}$ with $n=10,15,\dots,50$, produced by different models. We used four types of simulated data:

 \begin{itemize}
     \item[1)] Hamming distance matrices between $n$ binary sequences
    of length $m$. In each experiment, a phylogenetic tree topology
    was generated using the coalescent model \cite{kingman1982genealogy},
    and sequences were simulated along its branches with substitution
    rate $p$ under the Camin--Sokal model \cite{Camin1965-eh}.  

     \item[2)] Euclidean distance matrices for $n$ points randomly sampled from $\mathbb{R}^m$.

     \item[3)] Random doubly stochastic matrices (RDSM) with zero diagonals, publicly available from \cite{catanzaro2025new}.

     \item[4)] Random integer matrices (RIM), also used in \cite{catanzaro2025new}.
 \end{itemize}

In both 1) and 2) we used $m=1000$, roughly matching the order of magnitude typical for targeted amplicon or single-cell sequencing \cite{lin2024high,kozlov2022cellphy}.

For empirical data validation we used the PhyloBench dataset \cite{spirin2024phylobench}. It contains 5,847 benchmark instances, each with $n=15,30,$ or $45$ sequences from
orthologous protein domains randomly sampled from several thousand orthologous groups across 12 taxonomic groups spanning Bacteria, Archaea, and Eukaryota.

The following algorithms were used for comparison:

\begin{itemize}
\item[1)] The Neighbor-Joining (NJ) algorithm \cite{Saitou:MBE:1987}, which, as noted above, is an agglomerative heuristic for the
BME problem \cite{gascuel2006neighbor}.

\item[2)] Methods from the FastME, a widely used phylogenetic inference tool based on the BME model \cite{lefort2015fastme} and tree rearrangement-based local search. We used two variants of
its SPR-based local search, initialized with either a random phylogeny or an NJ-reconstructed phylogeny.
\end{itemize}

\subsection{Computational experiments}

\begin{figure}[ht!]
    \centering
    \includegraphics[width=\textwidth]{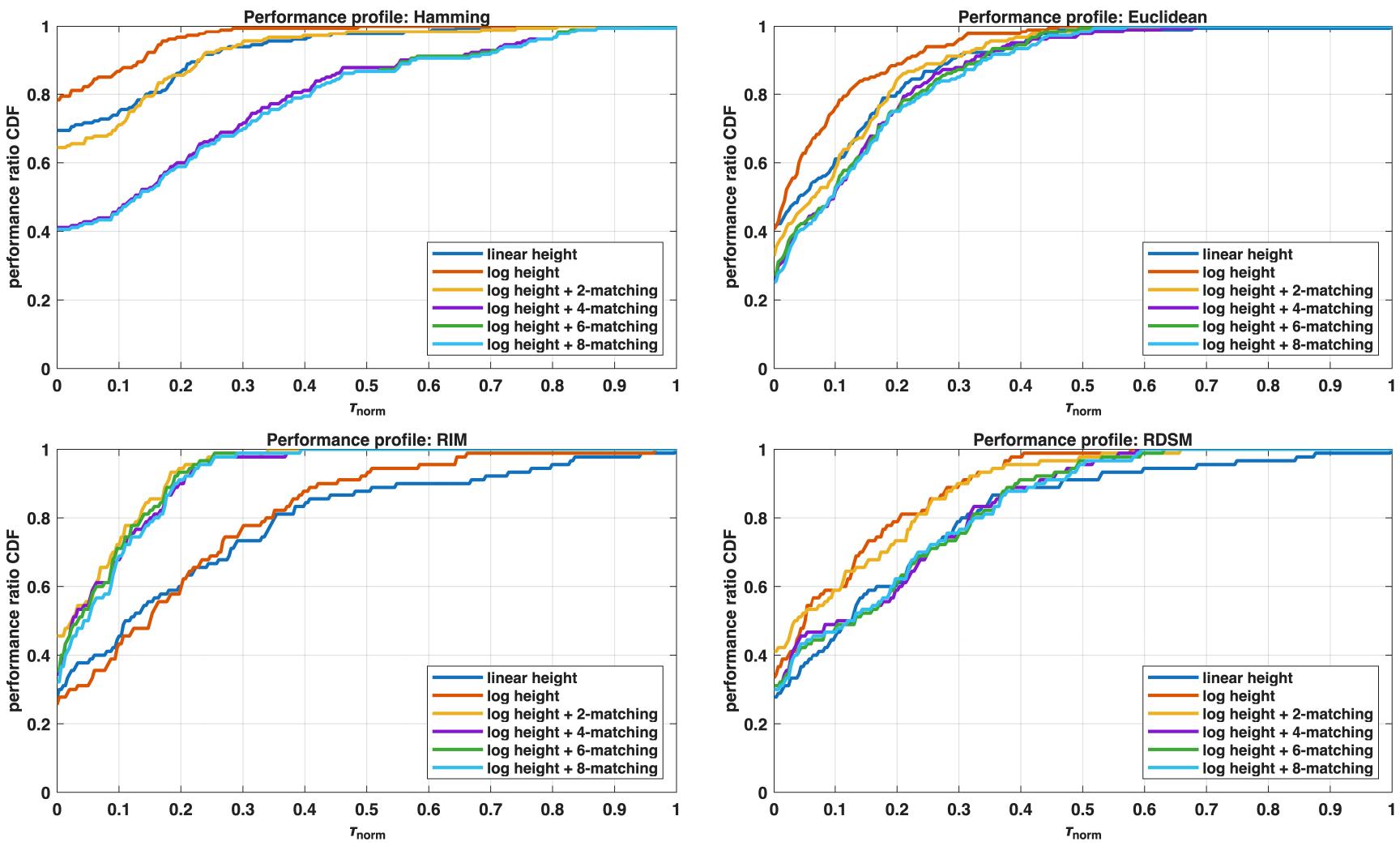} 
    \caption{Performance profiles of SDP relaxation and rounding under different algorithm settings. X-axis represents normalized performance ratio values $\tau_{norm} = \frac{\tau-1}{\tau_{\max}-1}$, $1\leq \tau \leq \tau_{\max}$.}
    \label{fig:algComparParamProf}
\end{figure}

In the first series of experiments, we compared different
parameterizations of the SDP relaxation and rounding
procedures (prior to SPR adjustment). We considered both
linear ($K = n/2$) and logarithmic ($K = 2\log n$) upper
bounds on the tree height in the SDP formulation. In the
logarithmic setting, we further evaluated different
matching sizes ($L = 1,2,4,6,8$) used during the merging
stage.

To compare algorithm performance, we used Dolan--Moré
performance profiles \cite{dolan2002benchmarking}
(Figs.~\ref{fig:algComparParamProf} and S1), defined as
follows. Given a set of solvers $\mathcal{S}$ and test
instances $\mathcal{D}$, the performance ratio of solver
$S\in\mathcal{S}$ on instance $D\in\mathcal{D}$ is
\[
r_{SD}=\frac{F_{SD}}{\min_{S'\in\mathcal{S}} F_{S'D}},
\]
where $F_{SD}$ is the objective value returned by $S$.
The performance profile
\[
\Phi_S(\tau)=\frac{1}{|\mathcal{D}|}
\sum_{D\in\mathcal{D}} \mathbf{1}\{r_{SD}\le\tau\}
\]
is the cumulative distribution function of these ratios
over $\mathcal{D}$. To compare profiles quantitatively, we use the area under
the curve (AUC)
\begin{equation}
\text{AUC}_S=\frac{1}{\tau_{\max}-1}
\int_{1}^{\tau_{\max}} \Phi_S(\tau)\,d\tau,
\end{equation}

\noindent
where $\tau_{\max}$ is the maximum ratio across all
solvers (Table~\ref{table:AUC}). Note that AUC here measures relative, not absolute, performance.

\begin{table}[]
\begin{center}
\begin{tabular}{|l|l|l|l|l|l}
\cline{1-5}
& \textbf{Hamming} & \textbf{Euclidean} & \textbf{RIM} & \textbf{RDSM} &  \\ \cline{1-5}
Linear height & 0.933            & 0.895              & 0.793        & 0.816         &  \\ \cline{1-5}
Log height & 0.972            & 0.934              & 0.814        & 0.895         &  \\ \cline{1-5}
Log height + 2-matching & 0.928            & 0.898              & 0.938        & 0.887         &  \\ \cline{1-5}
Log height + 4-matching & 0.799            & 0.875              & 0.929        & 0.838         &  \\ \cline{1-5}
Log height + 6-matching & 0.792            & 0.877              & 0.932        & 0.836         &  \\ \cline{1-5}
Log height + 8-matching & 0.791            & 0.870              & 0.925        & 0.835         &  \\ \cline{1-5}
\end{tabular}\caption{AUC for performance profiles of SDP relaxation and rounding under different algorithm settings.}\label{table:AUC}
\end{center}
\end{table}

We observed that a conservative linear upper bound on
tree depth, in addition to being more computationally
intensive, generally leads to lower accuracy than a
logarithmic bound. This can be explained by the fact that
bounding the depth $K$ in the SDP relaxation acts as a form of regularization: smaller $K$ limits how small the dyadic
weights $\zeta_i$ can become, reducing the model’s dynamic
range and improving conditioning. It also restricts the
relaxation’s ability to hide inconsistencies by pushing
leaves arbitrarily deep, while fewer levels limit the
accumulation of slack and relaxation looseness from
nestedness and balance constraints, yielding solutions
that are more coherent and easier to round. It also should be noted that limiting $K$ improves relaxation and
rounding, while still allowing recovery of deeper or less
balanced trees during agglomeration when warranted by the data.

The effect of the matching size $L$ (the number of leaves
merged at each agglomeration step) is less predictable
and more data-dependent, as merging multiple cherries can
have both positive and negative effects. When the SDP
produces well-separated candidates, merging several cherries in one step is beneficial: it performs multiple contractions jointly and is largely order-insensitive, avoiding a common failure mode of sequential greedy merging, where an early suboptimal choice prevents a later correct merge. Conversely, when cherry scores are ambiguous (i.e., have small margins or competing partners), merging many at once can reduce accuracy by committing to several uncertain contractions before re-solving the SDP. This limits the effect of the updated input and can amplify the impact of incorrect merges.

\begin{figure}
    \centering
    \subfigure{\includegraphics[width=0.49\textwidth]{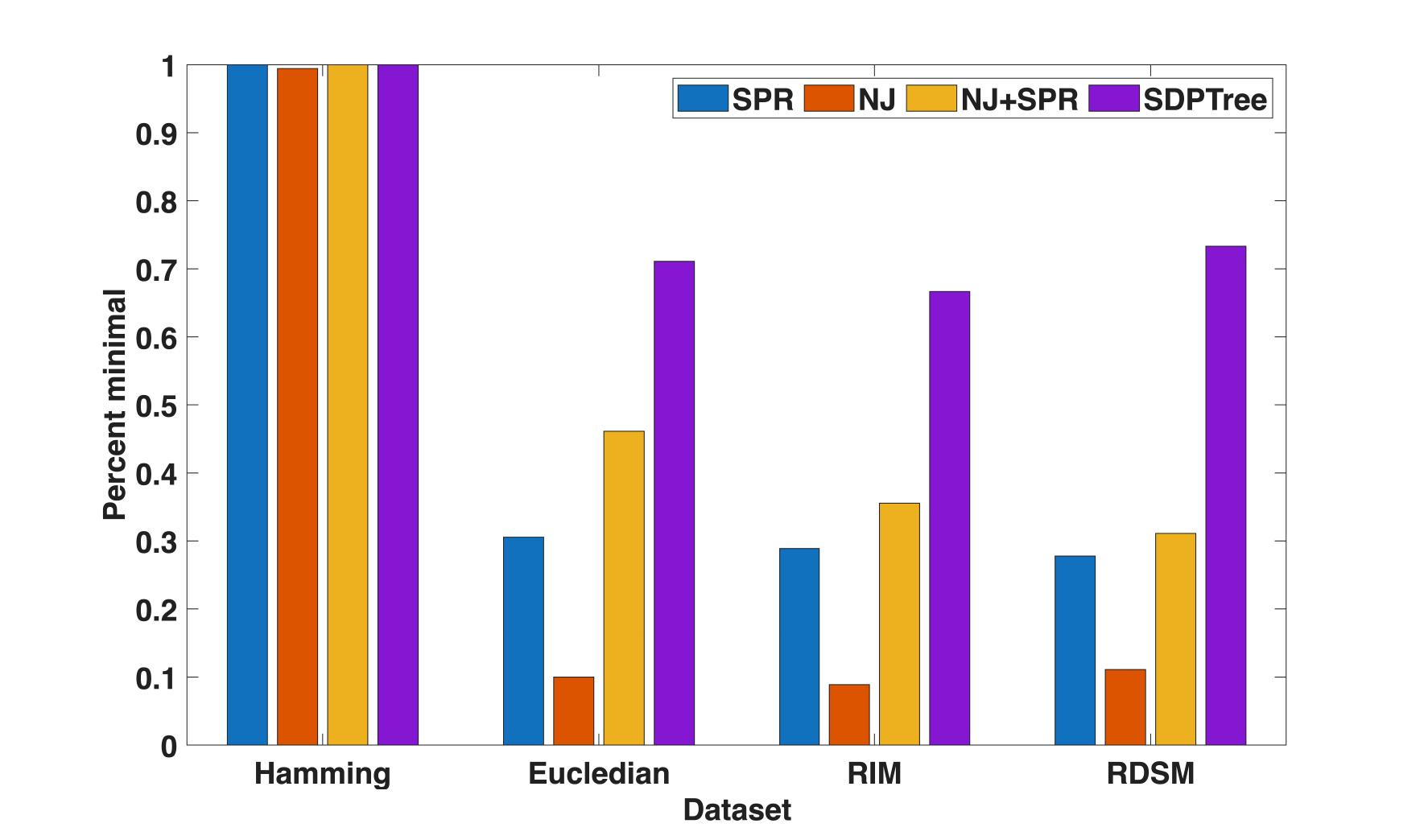}} 
    \subfigure{\includegraphics[width=0.49\textwidth]{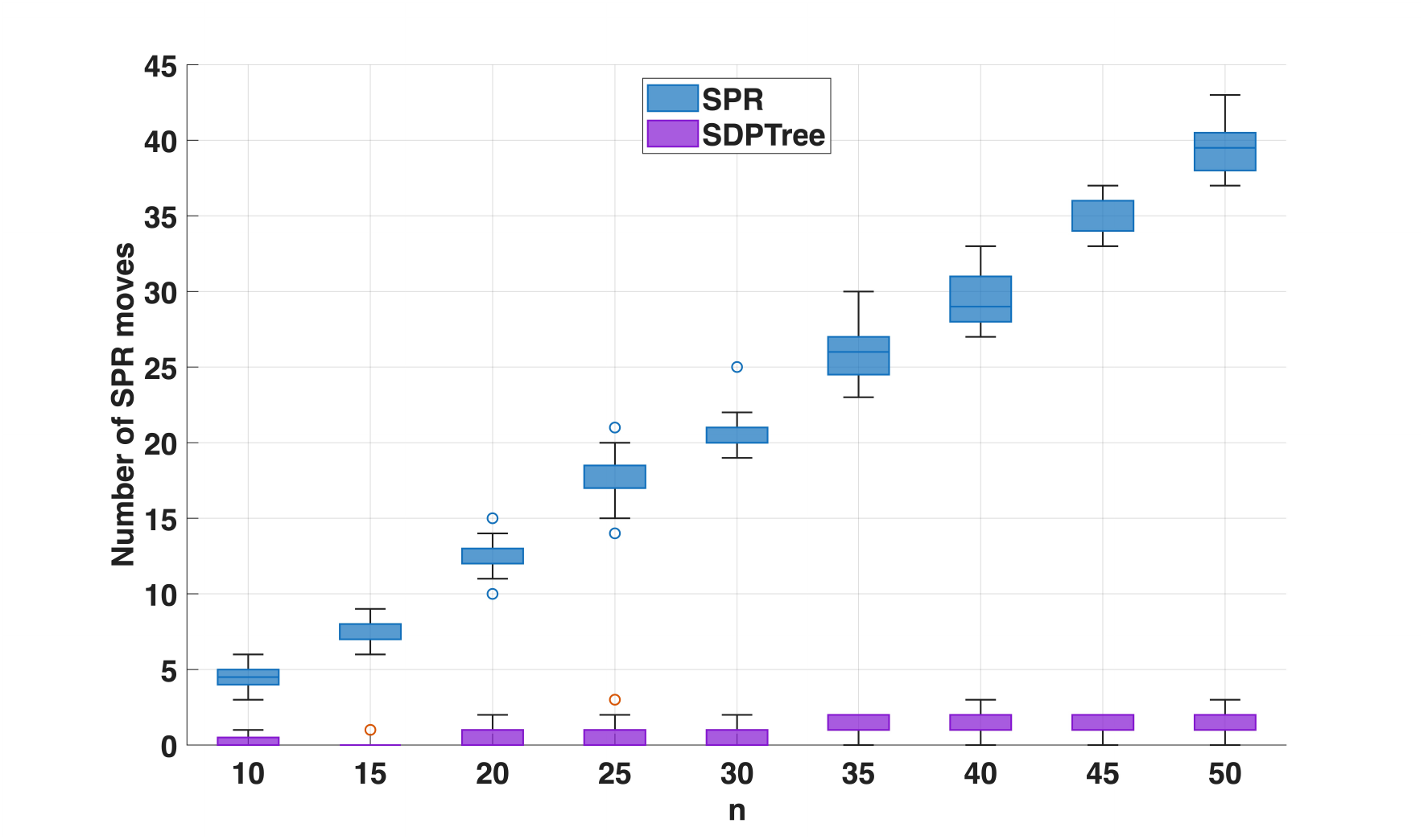}} 

    \caption{\footnotesize {\bf Left:} Joint comparison of four methods on simulated data. The y-axis shows the percentage of instances where a method achieved the best objective value among the four methods (ties allowed). {\bf Right:} Number of SPR moves for FastME (random initial tree) and SDPTree after SDP rounding for the Hamming dataset}
    \label{fig:algComparSim}
\end{figure}

The above relationships between algorithm settings are
largely consistent across input sizes (Figs.~S3--S5). The
matching size has a more pronounced effect for
$p$-rounding (Figs.~S1--S2), whereas for $s$-rounding
the number of maximally inseparable leaf pairs tends
to saturate and thus remains unchanged for large $L$.


In the next series of experiments, we compared SDPTree
with other methods. Based on the above results and the
trade-off between solution quality and running time (Fig.~\ref{fig:real}, right),
we selected the logarithmic depth bound with $L=2$
for benchmarking. Overall, SDPTree consistently
outperformed competing methods across
simulation settings (Fig.~\ref{fig:algComparSim}).
In pairwise comparison with the second-best method
(SPR with an NJ initial tree), SDPTree produced
better, identical, and worse solutions in $58.6\%$,
$22.5\%$, and $18.9\%$ of cases, respectively
($p<10^{-13}$, two-sided sign test). The advantage
of SDPTree tends to increase with input size
(Fig.~S6).

The exception is the Hamming distance dataset, where all
four algorithms almost always returned identical objective
values; this dataset was specifically designed to verify such behavior. Here, the
substitution rate was set to $p=0.05$, yielding nearly
additive matrices $D$ (with quartet splits matching the
generating tree in $99.2\%$ of cases on average). 
It is known that in such settings NJ with and without SDP adjustment produces optimal or near-optimal
solutions \cite{atteson1999performance}. The experiments
confirm that SDP relaxation with rounding preserves this guarantee. Indeed, SDP rounding landed very close to the optimal solution without
SPR in almost all cases: refinement required at most 2 SPR moves in $96.8\%$
of cases, including 0 moves in $39.4\%$ and 1 move in
$42.2\%$. The number of moves was nearly independent of
input size and significantly lower than for SPR with a
random initial tree ($p<10^{-40}$, two-sided sign test;
Fig.~\ref{fig:algComparSim} right).

A similar pattern was observed for Phylobench data
(Fig.~\ref{fig:real}). Across all three problem sizes, the SDPTree consistently outperforms the benchmark baselines, achieving the best objective value on the largest share of test instances and maintaining this advantage as $n$ increases ($p<10^{-84}$, two-sided sign test for SDPTree vs other algorithms). Overall, these results indicate that the SDP relaxation captures informative global structure that translates into high-quality trees after rounding, and, from the optimization perspective, the resulting solutions are  superior to standard heuristics and local search baselines across the tested sizes.


\begin{figure}
    \centering
    \subfigure{\includegraphics[width=0.49\textwidth]{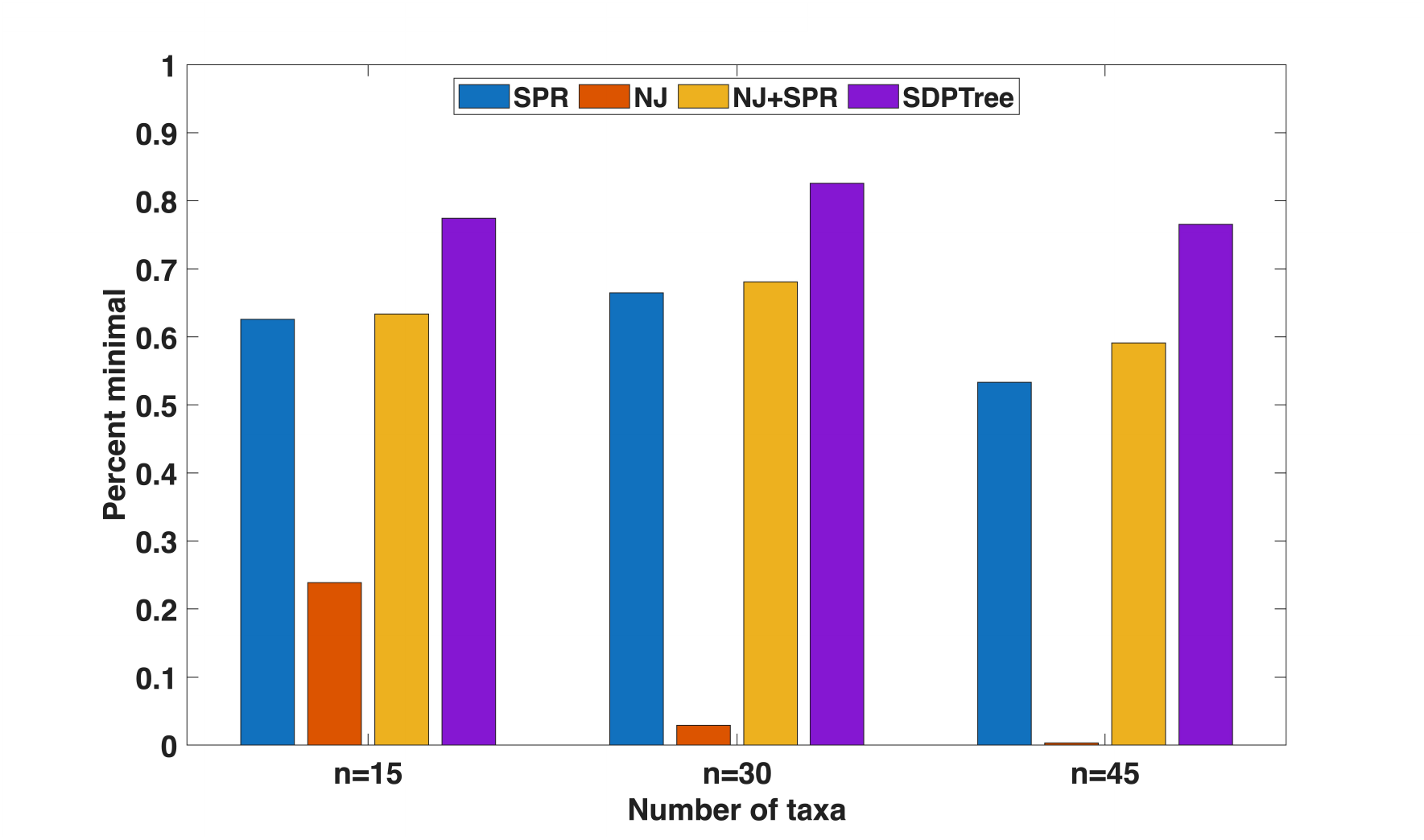}} 
    \subfigure{\includegraphics[width=0.49\textwidth]{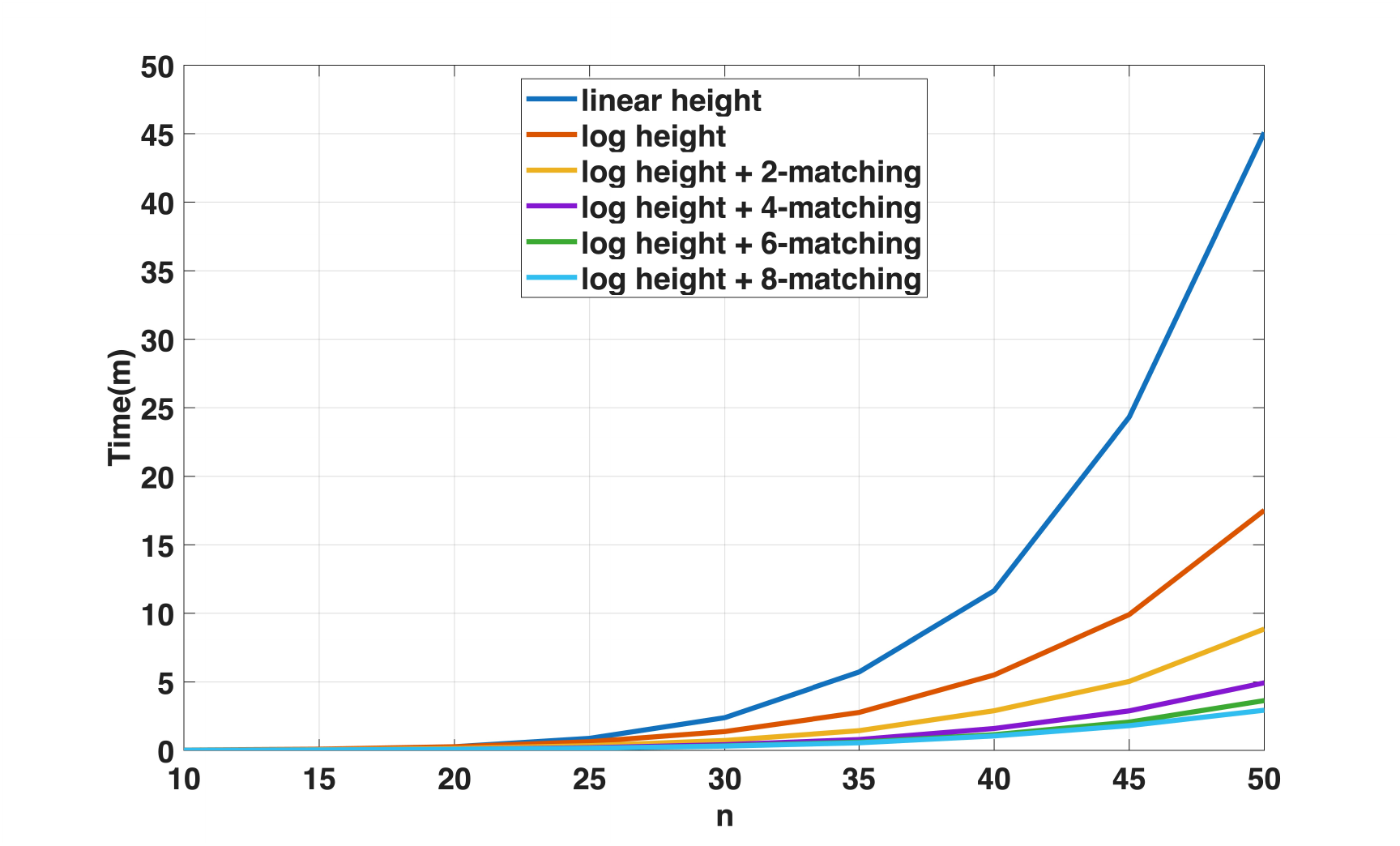}} 

    \caption{\footnotesize {\bf Left:} Joint comparison of four methods on Phylobench data. The y-axis shows the percentage of instances where a method achieved the best objective value among the four methods (ties allowed). {\bf Right:} running time of the single SDP relaxation + rounding step (in minutes)}
    \label{fig:real}
\end{figure}

\section{Discussion}

In this study, we proposed a novel algorithmic approach to phylogenetic inference based on Semidefinite Programming (SDP) and applied it to the Balanced Minimum Evolution (BME) problem,
one of the most widely used distance-based models. To the best of our knowledge, this is among the first attempts
to harness SDP, that has previously been successful in combinatorial optimization, for computational phylogenetics. 

The semidefinite relaxation strategy based on cut-based constraints developed here can potentially be transferred to other phylogenetic optimization problems. Many objectives of interest in phylogenetics can be expressed (exactly or approximately) as linear or convex functions of pairwise quantities induced by a tree, while the main combinatorial difficulty lies in enforcing that the underlying representation is consistent with a valid tree topology. In such settings, one can introduce lifted matrix variables encoding co-membership or split structure, impose tractable convex constraints capturing necessary properties of trees (e.g., hierarchical nesting/laminarity, balance or capacity conditions, and metric or ultrametric inequalities), and then tighten the relaxation by adding valid linear or conic cuts. This paradigm suggests SDP-based bounding and rounding frameworks for related minimum-evolution variants, least-squares tree fitting, quartet-consistency formulations, and ultrametric (molecular clock) models, where the relaxation provides both certificates of optimality gaps and informative fractional structure that can be exploited by downstream rounding procedures.

There is, however, room for improvement. The main limitation of the SDP approach, as with many mathematical programming methods, is running time and memory consumption: the relaxation introduces multiple positive semidefinite matrix variables, making generic interior-point methods expensive. Moreover, the relaxation may remain loose in regimes where the data are noisy or the depth/topology variables provide many degrees of freedom, leading to fractional solutions whose hierarchical signal is diffuse and therefore harder to round reliably. Several directions could improve both tightness and scalability: 

\begin{itemize}
    \item[(i)] developing stronger, problem-specific cuts (e.g., additional laminarity/compatibility constraints, objective-aware inequalities, and tighter envelope constraints linking matrix variables), so that tighter bounds can be obtained without uniformly inflating the model size;

    \item[(ii)] exploiting warm starts and model reuse to accelerate repeated solves in agglomerative pipelines, including adaptive termination tolerances tuned to rounding quality;

    \item[(iii)]  replacing the agglomerative leaf-merging scheme by a single SDP solve followed by SDP-based randomized projection rounding (e.g. Goemans–Williamson–type hyperplane rounding) producing a sequence of bipartitions that can be assembled into an unrooted tree, possibly repeated to sample multiple candidate solutions. 
\end{itemize}

\noindent
Implementing these improvements, as well as extending the SDP-based framework to related phylogenetic inference objectives, are promising directions for future research.

\paragraph{Acknowledgements.} This work was supported by National Science Foundation [grants 2415564, 2415562]. 

\bibliographystyle{plain}
\bibliography{refs}

\end{document}